\documentclass[letterpaper]{article} 
\usepackage{aaai23}  
\usepackage{times}  
\usepackage{helvet}  
\usepackage{courier}  
\usepackage[hyphens]{url}  
\usepackage{graphicx} 
\usepackage{natbib}  
\usepackage{caption} 
\usepackage{algorithm}
\usepackage{algorithmic}

\usepackage{newfloat}
\usepackage{listings}
\DeclareCaptionStyle{ruled}{labelfont=normalfont,labelsep=colon,strut=off} 
\floatstyle{ruled}
\newfloat{listing}{tb}{lst}{}
\floatname{listing}{Listing}
\title{Mediated Cheap Talk Design (with proofs)}
\author{
   Itai Arieli,
   Ivan Geffner,
   Moshe Tennenholtz\thanks{The work by Ivan Geffner and Moshe Tennenholtz was supported by funding from
the European Research Council (ERC) under the European Union’s Horizon 2020
research and innovation programme (grant agreement 740435).}
}
\affiliations{
    Technion - Israel Institute of Technology
}

\usepackage[centertags]{amsmath}
\usepackage{amsfonts,amsthm}
\usepackage{amsmath}
\usepackage{graphicx}
\newcommand{\commentout}[1]{}
\usepackage{booktabs}
\usepackage{appendix}
\usepackage{etoolbox}
\usepackage{tikz}
\usepackage{soul}
\usepackage{comment}
\usepackage{setspace}
\usepackage{xcolor,framed}
\usepackage{enumitem}
\usepackage{soul}   
\setstcolor{red}

\usepackage{threeparttable}
\usepackage{array}
\usepackage{booktabs}

\usepackage{blindtext}

\theoremstyle{definition}

\newcommand{\real}{\mathbb{R}}

\newcommand{\p}{\mathbf{P}}

\newtheorem{lemma}{Lemma}

\newtheorem{theorem}{Theorem}

\newtheorem{corollary}{Corollary}

\newtheorem{definition}{Definition}
\newtheorem{example}{Example}
\newtheorem{exer            cise}[theorem]{Exercise}

\newtheorem*{lem:main}{Lemma \ref{lem:main}}

\newtheorem*{cor:int}{Corollary \ref{cor:int}}

\begin{document}

\maketitle

\begin{abstract}
We study an information design problem with two informed senders and a receiver in which, in contrast to traditional Bayesian persuasion settings, senders do not have commitment power. In our setting, a trusted mediator/platform gathers data from the senders and recommends the receiver which action to play. 
We characterize the set of implementable action distributions that can be obtained in equilibrium, and provide an $O(n \log n)$ algorithm (where $n$ is the number of states) that computes the optimal equilibrium for the senders.  Additionally, we show that the optimal equilibrium for the receiver can be obtained by a simple revelation mechanism. 

\end{abstract}





\section{Introduction}
The 
extensive
literature on information design and Bayesian persuasion
studies optimal information revelation policies for the informed player. The two leading models of information revelation are cheap talk \cite{crawford1982strategic} and Bayesian persuasion \cite{kamenica2011bayesian}. The main distinction 
between
these models is the underlying assumption
that in the Bayesian persuasion models the sender has commitment power in the way
she discloses the information. 

Commitment power in the Bayesian persuasion model is crucial 
(see, e.g., \cite{Conitzer2006-ur})
and, while it may hold in some real-world settings,\footnote{The leading motivating example of \cite{kamenica2011bayesian} is of a prosecutor persuading a judge.} it is often considered strong. Another fundamental assumption in Bayesian persuasion models is that the informed player is also the one that designs the information revelation policy. In  practice, however, information revelation can be determined by other external or legal constraints. For example, information revealed to a potential customer about a product is determined by the commerce platform based on information submitted by different suppliers.  

On the other hand, in the cheap talk model there is no commitment power. However, in many cases, this lack of commitment leads to a lack of expressive power from the 
sender and, as a result, may 
induce
highly inefficient outcomes. For example, consider the classic market for lemons example in \cite{akerlof1978market},
in which a marketer tries to sell a product to a customer 
that
plays the role of the receiver. The product can be of either good quality or bad quality with equal 
probability.
The customer would like to buy the product only if he believes that it is of good quality with a probability of at least $\frac{3}{4}$, and the sender 
always prefers that the product is bought.
In 
this
case, under a cheap talk equilibrium, the sender has no way to signal credible information to the receiver and the receiver never buys the product.  
This is a typical situation that arises in marketplaces 
that match sellers with buyers, or advertisers with consumers.   

In our model there is a finite state space of size $n$, two informed players (senders), and an additional uninformed player (the receiver) that determines the outcome of the game by playing a binary action from the set $A := \{0,1\}$ (this could represent buying a product or not, passing a law or not, etc.). The utility of each player is determined by the game state and by the action played by the receiver, and the incentives of the senders may not necessarily be aligned (e.g., senders can be a car seller and  a technician that tested the car, or two parties who studied  the monetary value of a law, or two suppliers of a product, etc.).
The state of the game is drawn from a prior distribution that is common knowledge among the players, but only the senders know its exact value. Thus, the senders' purpose is to reveal information to the receiver in such a way that the receiver plays the action that benefits them the most. Since the senders have no commitment power we are interested in cheap talk equilibria, in which it is never in the interest of the senders to be dishonest, and it is always in the interest of the receiver to play the action suggested by the protocol. As we show in this work, the existence of a second informed sender dramatically enriches the set of cheap talk equilibria that can be obtained.

We consider a \emph{mediated cheap talk} setting of communication between the senders and the receiver. In this setting, the senders communicate with a trusted mediator, and as a function of the two messages that the mediator receives, he sends an action recommendation (possibly at random) to the receiver. 
Our first result provides a characterization 
of
the 
truthful equilibria that are implementable in the mediator setting,
in which
it is always beneficial for the senders to report truthful information and for the receiver to play whatever is suggested by the mediator.
We then analyze
the case where the two senders have aligned preferences and provide an algorithm with $O(n\log n)$ steps to calculate the best equilibrium outcome and payoff for the senders.
We later extend this algorithm to find the optimal outcome for one of the senders whenever the senders have different incentives.
Finally, we study the best equilibrium for the receiver and show that the optimal revelation policy lies within a finite set of 
mechanisms.

A major motivation for our work is data-driven decision making. Recommendation systems and classifiers are at the heart of many systems and determine the offering for or grouping of users based on data collected and provided by data sources. While in the early days these systems were based on data aggregated from their users and from previous interactions with them, the explosion of data facilitated professional data aggregation, and systems are designed to work with external data sources. Needless to say, data sources may be strategic and may attempt to influence the system's decisions. In abstract terms, the system acts as a mediator aiming in implementing a policy, that is a mapping from a state (e.g., type of user) to an action (recommendation, group assignment) based on messages received by data sources that can access the state. In general there may be several data sources, each of which has access to the required data but they have different preferences regarding the policy to be 
implemented by the system. 
Therefore, the main theoretic question is 
to know
which policies can be implemented in the strategic game between the data providers. Given that knowledge, we can tackle the question of what would be a mechanism that maps data sources' messages to actions
that are
optimal when the system aims to implement a particular policy. Our work provides rigorous answers to the 
above
question. This is complementary to work exploiting commitment power in data-intensive tasks, such as segmentation (e.g. \cite{EmekFGLT12}) and incentive-compatible exploration and exploitation \cite{KremerMP13,BaharST19}; in our setting, information providers do not have commitment power. 

Finally, as shown in \cite{ADGH06}, the assumption of communicating with a trusted mediator in most cases can be replaced with the assumption that both the senders and the receiver can communicate via private authenticated channels. This is true as long as (a) there exists a punishment strategy for the senders/receiver, or (b) we allow an arbitrarily small probability of error. This means that if the receiver or the senders can punish other participants when they are caught deviating (e.g., by quitting the game if all outcomes give positive utilities to the players), the same equilibria that can be achieved in the mediator setting can also be achieved in a cheap talk equilibrium without a mediator. If there is no such punishment strategy, the sets of equilibria in both settings might not be equal, but for any equilibrium in the mediator setting, there exist equilibria in the unmediated setting that are arbitrarily close.
\subsection{Related Literature}

The literature on information design is too vast to address all the related work. We will therefore mention some key related papers. The work by \cite{krishna2001model} considers a setting that is similar to the one considered by \cite{crawford1982strategic}, where a real interval represents the set of states and actions. In this setting the receiver's and the senders' utilities are \emph{biased} by some factor  that afects their incentives and utility. In \cite{krishna2001model} there are  two informed senders that reveal information sequentially to the receiver. They consider the best receiver equilibrium and show that, when both senders are \emph{biased} in the same direction, it is never
 beneficial to consult both of them. By contrast, when senders are biased in opposite directions, it is always beneficial to consult both of them. Our setting is different than theirs as we consider a finite state space and a binary action set for the receiver. In addition, our focus is the best equilibrium for either both or one of the senders. 
 
In another work \cite{salamanca2021value} characterizes the optimal mediation for the sender in a sender-receiver game. Relatedly,
\cite{lipnowski2020cheap}, and \cite{kamenica2011bayesian} provide a geometric characterization of the best cheap talk equilibrium for the sender under the assumption that the sender's utility is state-independent. 
In \cite{Gan2022-gm} and \cite{Fujii2022-aa}, the authors study the complexity of finding equilibria in sequential decision-making settings and in settings where the receiver's actions are specified by combinatorial constraints, respectively.

\cite{kamenica2017competition} consider a setting with two senders in a Bayesian persuasion model. The two senders, as in the standard Bayesian persuasion model, have commitment power and they compete over information revelation.
The authors characterize the equilibrium outcomes in this setting.

Integrating mediators into a strategic setting is common in many game-theoretical works (e.g., \cite{aumann1987correlated}, \cite{morgan1999models}). In more recent work, \cite{kosenko2018mediated} and \cite{arieli2022bayesian} study mediators in a Bayesian persuasion model. In these works the mediators are strategic players that may affect information revelation to the receiver. By contrast, in this work we remain agnostic to the incentives of the mediator and, as in \cite{aumann1987correlated}, the mediator only serves as a correlation device.

\section{The Model}

Throughout the rest of the paper we will focus on the mediator setting since the mechanisms involved are much more simple than those required for the cheap talk setting, since the latter require non-trivial distributed computing primitives. Fortunately, as shown in~\cite{ADGH06}, all results obtained in the mediator setting also apply in the cheap talk setting, except for an arbitrarily small probability of error.

We start by suppressing the dependence of feasibility in the receiver preferences and instead only study the mechanism with two players. 
Consider a finite state space $\Omega$ with a common prior $\mu\in\Delta(\Omega)$. There are two players that will later play the role of the senders but can also be considered as two political parties. There is a bill that can be either approved or rejected. The utilities of the two players are $u_1,u_2:A\times\Omega\to\real$, where $A=\{0,1\}$. We assume that both players observe the realized state.  We call $A$ the action set. 

Define a communication protocol $(M_1,M_2,\tau)$ that is implemented by a mediator as follows. For $i=1,2$ the set $M_i$ is the finite message space of player $i$. The function $\tau:M_1\times M_2\to\Delta(A)$, which is implemented by the mediator, maps a pair of messages to a probability over the action $a\in A$ (although, for simplicity, for the rest of the paper we will associate $\tau(m_1, m_2)$ with the probability that the mediator suggests $0$). A communication protocol defines a game between the two players, where a (behavioral) strategy of player $i$  is a mapping $\sigma_i:\Omega\to\Delta(M_i)$. A communication protocol $(M_1,M_2,\tau)$ together with a  profile of strategies $\sigma=(\sigma_1,\sigma_2)$ induces a probability measure $\p_\sigma\in(A\times\Omega)$. 

We call a Nash equilibrium of the game induced by a communication protocol a \emph{cheap talk equilibrium}. 
A \emph{policy} is a mapping $p:\Omega\to\Delta(A)$. For simplicity we identify $p(\omega)$ with the probability of recommending action $a=0$ in state $\omega$. Say that a policy $p$ is \emph{cheap talk implementable} (or simply implementable) if there exists a communication
protocol $(M_1,M_2,\tau)$ and a corresponding cheap talk equilibrium $\sigma$ such that $\p_\sigma(a=0|\omega)=p(\omega)$ for every $\omega\in\Omega$.

Our first main goal is to characterize the set of implementable policies $p$. To do this, we define a binary relation $\prec$ over $\Omega$ as follows.
\commentout{
For two distinct $\omega,\omega'$ we have that $\omega'\prec \omega$  iff $(u_i(0,\omega)-u_i(1,\omega))(u_j(0,\omega')-u_j(1,\omega'))<0$ for $\{i,j\}=\{1,2\}.$ 
That is,  $\omega'\prec \omega$ holds if one of the senders strictly prefers action zero at $\omega$ and the other strictly prefers action $1$ at $\omega'$. 
Note that $\prec$ is not necessarily symmetric. To see this assume that both players prefer action $0$ to $1$ at $\omega$ and action $1$ to $0$ at $\omega'$. Thus we have that $\omega'\prec \omega$ but not vice versa.
}
For two states $\omega, \omega'$ we define the relation $\omega' \prec \omega$ iff $u_i(0, \omega) > u_i(1, \omega)$ and $u_j(0, \omega') < u_j(1, \omega')$ for $i,j \in \{1,2\}$ and $i \not = j$. That is  $\omega'\prec \omega$ holds if one of the senders strictly prefers action $0$ at $\omega$ and the other strictly prefers action $1$ at $\omega'$. 
Note that $\prec$ is not necessary symmetric. To see this assume that both players prefers action $0$ to $1$ at $\omega$ and action $1$ to $0$ at $\omega'$. Thus we have that $\omega'\prec \omega$ but not vice versa.

Our characterization goes as follows
\begin{theorem}\label{th:main}
$p:\Omega\to\Delta(A)$ is an implementable policy iff $p(\omega')\leq p(\omega)$ for every pair $\omega,\omega'\in\Omega$ such that $\omega'\prec \omega$. 
\end{theorem}

Surprisingly, the set of implementable policies is independent of the prior $\mu$.
We note that in the case where the two players are never indifferent between the outcomes in $A$, our characterization gets even simpler form. Specifically, the set $\Omega$ can be partitioned into four kinds of states: $\Omega_{1,1}\subseteq\Omega$, where both players prefer action $1$; $\Omega_{0,0}$, where both prefer action $0$; $\Omega_{1,0}$, where player $1$ prefers action $1$ and player 2 prefers action $0$; and $\Omega_{0,1}$, where player $1$ prefers action $0$ and player $2$ prefers action $1$. 

\begin{corollary}\label{corollary:no ties}
In cases where indifference never holds, a policy $p:\Omega\to[0,1]$ is implementable iff the following conditions hold:
\begin{enumerate}
    \item [(1)] The minimum of $p$ among all  states in $\Omega_{0,0}$ is (weakly) larger than the maximum of $p$ among all other states.
    \item [(2)] The maximum of $p$ among all states in $\Omega_{1,1}$ is (weakly) larger than the minimum of $p$ among all other states.
    \item [(3)] $p$ is constant on $\Omega_{1,0}$.
    \item [(4)] $p$ is constant on $\Omega_{0,1}$.
\end{enumerate}
\end{corollary}

\begin{proof}[Proof of Theorem~\ref{th:main}]
By the revelation principle, we can transform any cheap talk equilibrium into a \emph{truthful} cheap talk equilibrium, in which both players send the current state to the mediator. Thus, for simplicity, we can assume that $M_1 = M_2 = \Omega$, and that $\sigma_1 \equiv \sigma_2 \equiv Id_\Omega$. Suppose that player 1 prefers action $0$ in state $\omega$ and player 2 prefers action $1$ in state $\omega'$. Then, player $1$ cannot increase the probability that the receiver plays action $0$ by sending another state to the mediator, and player $2$ cannot decrease such probability. If we plot the values of $\tau(\cdot, \cdot)$ in an $n \times n$ matrix as in Figure~\ref{fig:matrix}, we get the following insight: $\tau(\omega, \omega)$ must be the maximum value in the $\omega$ column and, simultaneously, $\tau(\omega', \omega')$ must be the minimum value in the $\omega'$ row. Thus, it must hold that $\tau(\omega, \omega) \ge \tau(\omega', \omega) \ge \tau(\omega', \omega')$, and therefore that $p(\omega) \ge p(\omega')$. 

\begin{center}
    \begin{figure}[h!]
    \includegraphics[]{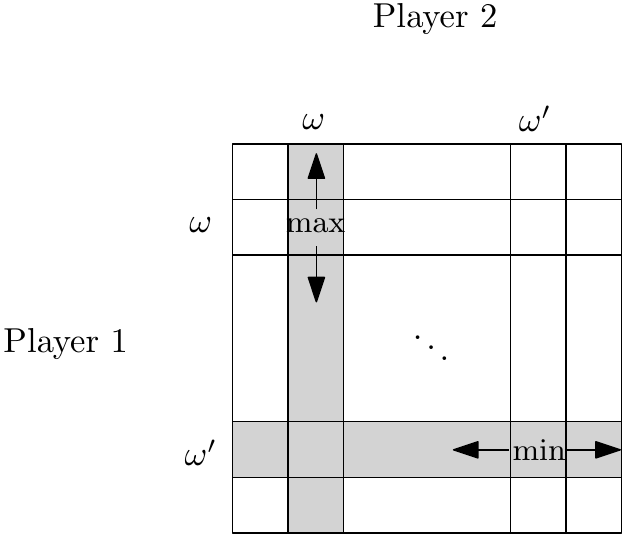}
    \caption{Matrix with the values of $\tau$}
    \end{figure}\label{fig:matrix}
\end{center}

Conversely, we can check that if $p(\omega) \le p(\omega')$ whenever $\omega \prec \omega'$, there is a way to construct $\tau$ such that $\tau(\omega, \omega) = p(\omega)$ for all $\omega \in \Omega$ and $(Id_\Omega, Id_\Omega)$ is a cheap talk equilibrium in $(\Omega, \Omega, \tau)$. The high-level idea is that the problem reduces to assigning a value to all entries in the matrix of Figure~\ref{fig:matrix} such that (a) $\tau(\omega, \omega)$ is always the greatest (resp., the smallest) entry in the $\omega$ column whenever player 1 prefers $0$ (resp., prefers 1), and (b) $\tau(\omega, \omega)$ is always the greatest (resp., the smallest) entry in the $\omega$ row whenever player 2 prefers $0$ (resp., prefers $1$). We can get an idea about how to fill this matrix by looking at Figure~\ref{fig:matrix}: if player 1 prefers $0$ in $\omega$, player 2 prefers $1$ in $\omega'$, and $p(\omega) \ge p(\omega')$, then we should simply assign a value to $\tau(\omega, \omega')$ such that $\tau(\omega, \omega) \ge \tau(\omega', \omega) \ge \tau(\omega', \omega')$, e.g., $\tau(\omega', \omega) := \frac{\tau(\omega, \omega) + \tau(\omega', \omega')}{2}$. The same value works if the preferences between player 1 and 2 are reversed. If both have the same preferences, we have that $\tau(\omega', \omega)$ should be either smaller or greater than both $\tau(\omega, \omega)$ and $\tau(\omega', \omega')$, which means that we can set $\tau(\omega', \omega) := 0$ or $\tau(\omega', \omega) := 1$ respectively. More precisely, let $\tau$ be such that 
\begin{enumerate}
    \item $\tau(\omega, \omega) = p(\omega)$ for all $\omega \in \Omega$.
    \item $\tau(\omega', \omega) = \frac{p(\omega)+p(\omega')}{2}$ if player 1 prefers 0 in $\omega$ and player 2 prefers $1$ in $\omega'$ or, vice-versa, if player 1 prefers 1 in $\omega'$ and player 2 prefers $0$ in $\omega$.
    \item $\tau(\omega', \omega) = 0$ if player $1$ prefers $0$ in $\omega$ and player 2 prefers $0$ in $\omega'$.
    \item $\tau(\omega', \omega) = 1$ if player $1$ prefers $1$ in $\omega$ and player 2 prefers $1$ in $\omega'$.
\end{enumerate}

If a player is indifferent between actions $1$ and $0$, we assume that she prefers $0$. Note that with this assumption, all possible cases are covered, and thus $\tau$ is defined for all pairs $(\omega', \omega)$. We next show that 
$(Id_\Omega, Id_\Omega)$ is a cheap talk equilibrium in $(\Omega, \Omega, \tau)$.

Suppose that the game state is $\omega$; we will show that neither player 1 nor player 2 can increase their utility by misreporting the state to the mediator. Suppose that player 1 prefers action $0$. If she were to lie and tell the mediator that the state is $\omega'$, there are two possibilities: if player 2 prefers action $1$ in $\omega'$, then $\omega' \prec \omega$ and thus, by construction (case 2), $\tau(\omega', \omega) = \frac{p(\omega)+p(\omega')}{2} \le \tau(\omega, \omega)$, where the last inequality derives from the fact that $\omega' \prec \omega \Longrightarrow p(\omega') \le p(\omega)$. This means that, in this case, player 1 cannot increase its utility by reporting $\omega'$ instead of $\omega$. If, instead, player 2 prefers $0$ in $\omega'$, we are in case 3 and thus $\tau(\omega', \omega) = 0 \le \tau(\omega, \omega)$ as before. If player 1 prefers $1$ in $\omega$, there are again two possibilities: if player 2 prefers $0$ in $\omega'$, then again $\tau(\omega', \omega) = \frac{p(\omega)+p(\omega')}{2}$, but in this case $\omega \prec \omega'$ and thus $p(\omega) \le p(\omega')$, which means that $\tau(\omega', \omega) \ge \tau(\omega, \omega)$ and, therefore, player 1 cannot increase its utility by reporting $\omega'$ instead of $\omega$. The remaining possibility is the one in which player 2 prefers action $1$ in $\omega'$. However, in this case we have that, by construction, $\tau(\omega', \omega) = 1 \ge \tau(\omega, \omega)$. An analogous argument shows that player 2 cannot increase her utility by misreporting (note that, by construction, $\tau$ is symmetric with respect to the preferences of player 1 and player 2).

\end{proof}
\section{Application to Information Design}
In this section we study the implication of
Theorem~\ref{th:main} 
for information design problems. In particular, in addition to the two informed players, who henceforth will be called senders, we have an uninformed receiver with a utility function $v:A\times\Omega\to\real$. In this setting,
the mediator does not play the action immediately, but instead suggests it to the receiver. The receiver plays the action if it is incentive-compatible to do so. More precisely, given protocol $(M_1, M_2, \tau)$ with strategy profile $(\sigma_1, \sigma_2)$, the mediator plays the action suggested by the mediator if and only if it gives the receiver a better expected utility than playing any other action. This is formalized in the definition below.

\begin{definition}
Given communication protocol 
$(M_1,M_2,\tau)$, a pair of behavioral strategies $\sigma=(\sigma_1,\sigma_2)$ induces a cheap talk equilibrium if,
for all $i \in \{1, 2\}$, the strategy $\sigma_i$ maximizes the utility of sender $i$ given $\tau$ and $\sigma_{-i}$, and, in addition, $\p_\sigma\in\Delta(A\times\Omega)$ induces an incentive-compatible recommendation for the receiver.
\end{definition}
Note that $\p_\sigma\in\Delta(A\times\Omega)$ induces an incentive-compatible recommendation for the receiver if and only if the following inequality holds:
$$\sum_{\omega\in\Omega}\p_\sigma(\omega|a)v(a,\omega)\geq \sum_{\omega\in\Omega}\p_\sigma(\omega|a)v(1-a,\omega).$$
That is, $\tau$ and $\sigma$ generate an action recommendation for the receiver. As is standard in the literature, such a recommendation is incentive compatible if, conditional on the recommendation on an action $a\in A$, the receiver is better off accepting the recommendation than playing action $1-a$. 
,
With this, we can define an \emph{implementable policy}.

\begin{definition}
A policy $p:A\to[0,1]$ is implementable if there exists a communication protocol 
$(M_1,M_2,\tau)$ and a cheap talk equilibrium $\sigma$
such that $\p_\sigma(a=0|\omega)=p(\omega)$ for every $\omega\in\Omega$. 
\end{definition}

Intuitively, $p$ is implementable if there exists a communication protocol and a cheap talk equilibrium that induce $p$.

\subsection{Common Interest among Senders}

One of the goals of this work is to characterize implementable policies.
We first consider the case where the two senders have a common interest. That is, $u_1=u_2=u.$

\commentout{
A natural comparison in this case is the
Bayesian persuasion setting of \cite{kamenica2011bayesian}.
In this setting, there is only one sender, but that sender has commitment power.

We recall that the prior probability on $\Omega$ is $\mu\in\Delta(\Omega)$. We let $\beta$ be the utility that the receiver can guarantee when he has no further information over $\Omega$. That is,
$\beta=\max_{a\in A}\sum_{\omega\in\Omega}\mu(\omega)v(a,\omega).$

We call a policy $p:A\to[0,1]$ implementable if there exists a communication protocol 
$(M_1,M_2,\tau)$ and a cheap talk equilibrium $\sigma$
such that $\p_\sigma(a=0|\omega)=p(\omega)$ for every $\omega\in\Omega$. 
}

Let $\beta$ be the utility that the receiver can guarantee when she has no information about the current state (i.e., $\beta = \max(E_\mu(v(1, \omega)), E_\mu(v(0, \omega))$, let $\Omega_0=\Omega_{0,0}$ and $\Omega_1=\Omega_{1,1}$, and, for each policy $p$, let $q_p\in\Delta(A\times\Omega)$ be the corresponding distribution that is generated by $p$ and the prior $\mu$. The
following lemma gives a first characterization of the implementable policies whenever both senders have the same utilities.
\begin{lemma}\label{lem:char}
A policy $p:\Omega\to [0,1]$ is implementable iff
$E_{q_p}[v(a,\omega)]\geq\beta$ and 
$\min_{\omega\in\Omega_{0}}p(\omega)\geq \max_{\omega\in\Omega_{1}}p(\omega)$
\end{lemma}
\begin{proof}
By Theorem~\ref{th:main}, the condition $\min_{\omega\in\Omega_0}p(\omega)\geq \max_{\omega\in\Omega_1}p(\omega)$  is necessary and sufficient for $p$ to be implementable by the senders.
Clearly, since information cannot harm the receiver, we must have $E_{q_p}[v(a,\omega)]\geq\beta$ for any implementable policy. We now show that the condition $E_{q_p}[v(a,\omega)]\geq\beta$ induces an incentive-compatible recommendation for the receiver. We note that
$$
\begin{array}{lll}
E_{q_p}[v(a,\omega)] & = & q_p(a=0)E_{q_p}[v(0,\omega)|a=0] \\
& + & q_p(a=1)E_{q_p}[v(1,\omega)|a=1].\\
\end{array}$$ If, by way of contradiction, the policy is not incentive compatible, then either $E_{q_p}[v(1,\omega)|a=0]>E_{q_p}[v(0,\omega)|a=0]$ or $E_{q_p}[v(0,\omega)|a=1]>E_{q_p}[v(1,\omega)|a=1]$ (or both). In any case, this means that always playing 0 or always playing 1 yields an expected payoff which is higher than $\beta$. This contradicts the definition of $\beta$.  
\end{proof}

This lemma shows that, for a policy to be implementable, (a) it has to satisfy that $0$ is always played with more probability whenever both senders prefer $0$ than when both senders prefer $1$, and (b) the expected utility of the receiver under this policy should always be better than when she receives no information at all. 

It is interesting to see how this setting compares with the one in which there is only one sender but with commitment power, as in \cite{kamenica2011bayesian}. It turns out that, in our setting, the set of implementable policies is more restrictive, and thus there are cases in which the optimal Bayesian persuasion mechanism with one sender with commitment power is not implementable, as the following example shows.

\begin{example}\label{example}

Consider a state space $\Omega=\{\omega_1,\omega_2,\omega_3,\omega_4,\omega_5,\omega_6\}$ with six states. The utility for the senders and the receiver is given in the following table: 
\begin{table}[H]
\begin{center}
\begin{tabular}{lcccccc}
                         & $\omega_1$                 & $\omega_2$                 & $\omega_3$                 & $\omega_4$                 & $\omega_5$                 & $\omega_6$                 \\ \cline{2-7} 
\multicolumn{1}{l|}{$v$} & \multicolumn{1}{c|}{$0,2$} & \multicolumn{1}{c|}{$1,0$} & \multicolumn{1}{c|}{$0,1.9$} & \multicolumn{1}{c|}{$0,2.1$} & \multicolumn{1}{c|}{$1,0$} & \multicolumn{1}{c|}{$2,0$} \\ \cline{2-7} 
\multicolumn{1}{l|}{$u$} & \multicolumn{1}{c|}{$0,1$} & \multicolumn{1}{c|}{$1,0$} & \multicolumn{1}{c|}{$1,0$} & \multicolumn{1}{c|}{$1,0$} & \multicolumn{1}{c|}{$0,1$} & \multicolumn{1}{c|}{$0,1$} \\ \cline{2-7} 
                         & \multicolumn{1}{l}{}       & \multicolumn{1}{l}{}       & \multicolumn{1}{l}{}       & \multicolumn{1}{l}{}       & \multicolumn{1}{l}{}       & \multicolumn{1}{l}{}      
\end{tabular}
\end{center}
\end{table}
For each state $\omega_i$, the left-hand number represents the utility from action $0$ and the right-hand number represents the utility from action $1$. The prior 
is the uniform 
distribution $\mu=(\frac{1}{6},\frac{1}{6},\frac{1}{6},\frac{1}{6},\frac{1}{6},\frac{1}{6})$. 
One can show that the optimal Bayesian persuasion policy generates the following conditional probability of action\footnote{The calculation is based on the standard  algorithm for computing the optimal persuasion policy for the sender in the case where the action set for the receiver is binary. See, e.g., \cite{arieli2019delegated} and \cite{renault2017optimal}.} $0$: $p=(0,1,1,0,0,0.45).$ Note that $\Omega_0=\{ \omega_2,\omega_3,\omega_4\}$ and $\Omega_1=\{ \omega_1,\omega_5,\omega_6\}$. Thus,  
$\min_{\omega\in\Omega_0}p(\omega)=0$ while $\max_{\omega\in\Omega_1}p(\omega)=0.45$.
Lemma \ref{lem:char} 
implies that $p$ is not implementable.
\end{example}

In the following section, we provide an efficient algorithm that finds the best equilibrium for the senders. Note that this is equivalent to finding the best implementable policy since we can always efficiently construct the protocol and the equilibrium that induces the resulting policy as in the proof of Theorem~\ref{th:main}.

\subsection{Best Sender Equilibrium}

\subsubsection{Common Interest among Senders}\label{sec:common-interest-algo}

In this section we show how to compute a best sender equilibrium when both senders have the same utilities. This will serve as a stepping-stone to the more general algorithm in Appendix~\ref{sec:general-case-sketch} that outputs the best equilibrium for the first sender, even in the case where senders don't have the same utilities. The algorithm provided runs in $O(n \log n)$ operations, where $n$ is the number of states in $\Omega$. Before we start, 
note that we can refine Lemma~\ref{lem:char} and get a better characterization of optimal implementable policies.

\begin{lemma}\label{lemma:beta-utility}
Let $p^s$ be the optimal policy for the senders (i.e., $p^s(\omega)=1$ for every $\omega\in\Omega_0$ and  $p^s(\omega)=0$ for every $\omega\in\Omega_1$), and let $\beta$ be the receiver's expected utility with no information. If $E_{q_{p^s}}[v(a,\omega)] \ge \beta$, then $p^s$ is the implementable policy that is optimal for the senders. Otherwise, all optimal implementable policies $p$ for the senders satisfy $E_{q_p}[v(a, \omega)] = \beta$.
\end{lemma}

\begin{proof}
If $E_{q_{p^s}}[v(a,\omega)] \ge \beta$, then it is also implementable by Lemma~\ref{lem:char} and, by construction, it is also optimal for the senders. Otherwise, by Lemma~\ref{lem:char}, any optimal implementable policy $p$ must satisfy $E_{q_p}[v(a,\omega)] \geq \beta$. Suppose that  $E_{q_p}[v(a,\omega)] > \beta$. By assumption, this means that $p$ is not equal to $p^s$ and, therefore, there exists either $\omega\in\Omega_0$ such that $p(\omega)<1$ or $\omega\in\Omega_1$ such that $p(\omega)>0$. Consider the former case. We can define a new policy $p'$ by 
setting $p'(\omega)=p(\omega)+\delta$ for some small $\delta>0$, and 
setting
$p'(\omega')=p(\omega')$ for 
all
other $\omega'\neq\omega$. 
By choosing a value of $\delta$ that is small enough, the inequality $E_{q_{p'}}[v(a,\omega)]\geq\beta$ is still satisfied.
In addition, since $p$ is incentive-compatible for the senders, so is $p'$. Therefore, $p'$ is implementable and yields a higher payoff than $p$ to the senders,
which contradicts the assumption that $p$ is optimal.
A similar construction can be applied when $p(\omega)>0$ for some $\omega\in\Omega_1$.
Hence, 
$E_{q_p}[v(a,\omega)]=\beta$, as desired.
\end{proof}

This lemma shows that we can restrict our search to implementable policies that give exactly $\beta$ utility to the receiver (in addition to the policy in which the receiver always plays according to the senders' preferences). We divide the rest of this section into two parts. First, we focus on a simpler problem that will be needed for the final algorithm and, second, we use the solution to this problem as a primitive for the final algorithm. The following notation will be useful.

Let $\Omega_C$ be the set of states where the senders and the receiver agree on the identity of the optimal action. That is, $\omega\in\Omega_C$ iff $[u(1,\omega)-u(0,\omega)][v(1,\omega)-v(0,\omega)]\geq 0$. Let $\Omega_D=\Omega\setminus\Omega_C$ be the set of disagreement states. We distinguish between $\Omega_{D,0}=\Omega_D\cap\Omega_0$ and $\Omega_{D,1}=\Omega_D\cap\Omega_1$. Let $\Omega_{C,0}$ and $\Omega_{C,0}$ be similarly defined.

Because of Theorem~\ref{th:main}, we know that any solution must satisfy that, for any $\omega \in \Omega_0$, $p(\omega)$ is greater than all $p(\omega')$ for $\omega' \in \Omega_1$. Given $\alpha \in [0,1]$, consider the problem of finding the best sender equilibrium that is constrained to $p(\omega) \ge \alpha$ for all $\omega \in \Omega_0$ and $p(\omega) \le \alpha$ for $\omega \in \Omega_1$. Let $p_\alpha$ denote the solution of this problem for a particular $\alpha$. By the above property, there exists an alpha such that $p_\alpha$ is the actual best  equilibrium for the senders (with no constraints).

Next, we show how to compute $p_\alpha$. Clearly, in the states $\omega$ in which the senders and the receiver have the same action  preference $a$, the mechanism should satisfy $p(\omega) = 1-a$ (i.e., $p(\omega) = 1$ whenever they all prefer 0 and $p(\omega) = 0$ whenever they all prefer 1). The main difficulty is finding the correct configuration for the disagreement states. Suppose that $\Omega_D = \{\omega_1, \ldots, \omega_k\}$ is sorted according to its states' resistance $r(\omega) := \frac{v(1,\omega) - v(0, \omega)}{u(0,\omega) - u(1, \omega)}$ from largest to smallest (note that these values are always positive). Consider the following algorithm that computes $p_\alpha$ (whenever it exists):

\begin{enumerate}
    \item \textbf{Step 1:} To each state $\omega \in \Omega$, assign $p(\omega) = 1$ if $\omega \in \Omega_0$ and assign $p(\omega) = 0$ otherwise. If the receiver's utility this way is larger than $\beta$, return this configuration and terminate.
    \item \textbf{Step 2:} Iterate through $\omega_1, \ldots, \omega_k$. If $p(\omega_i) = 1$, decrease this value until either $p(\omega_i) = \alpha$ or the receiver's utility equals $\beta$. If $p(\omega_i) = 0$, increase this value until either $p(\omega_i) = \alpha$ or the receiver's utility equals $\beta$. After each step, if the utility of the receiver is equal to $\beta$, return this configuration and terminate.
    \item \textbf{Step 3:} If no solution was found in Step 2 (which can only happen if, after iterating through all $k$ elements, the utility of the sender never reached $\beta$), there is no incentive-compatible configuration for $\alpha$.
\end{enumerate}


We claim that this protocol returns $p_\alpha$ if it exists. In fact, given the configuration at the end of Step 1, note that decreasing $p(\omega)$ by $\Delta$ for $\omega \in \Omega_{D,0}$ or increasing $p(\omega)$ by $\Delta$ for $\omega \in \Omega_{D,1}$ increases the receiver's utility by $\Delta |v(1, \omega) - v(0, \omega)|$ while it decreases the senders' utilities by $\Delta |u(1, \omega), u(0, \omega)|$. By Lemma~\ref{lem:char}, $p$ is implementable if and only if the receiver gets a utility of at least $\beta$. Therefore, in order for the receiver to get utility $\beta$ while minimizing the losses for the senders, it is optimal to decrease/increase as much as possible the values of $p$ with the highest resistance, as in Step 2. This shows that the algorithm above returns the correct solution.

Now that we are able to compute $p_\alpha$, it remains to find for what value of $\alpha$ the policy $p_\alpha$ maximizes the senders' utilities. We claim that we can restrict the search to values of $\alpha$ such that $p_\alpha$ attains values only in $\{0, \alpha, 1\}$.  

\begin{lemma}
There exists $\alpha \in [0,1]$ such that $p_\alpha$ is optimal for the senders and $p_\alpha(\omega) \in \{0, \alpha, 1\}$ for all $\omega \in \Omega_D$.
\end{lemma}

\begin{proof}
Suppose that there exists a solution $p_\alpha$ such that $j$ is the first index such that $p_\alpha(\omega_j) \not \in \{0, \alpha, 1\}$. By construction, the solution provided by the algorithm satisfies that $p_\alpha(\omega_i) = \alpha$ for all $i < j$, and for $i > j$ $p_\alpha(\omega_i)$ is $1$ or $0$ depending on whether $\omega_i \in \Omega_0$ or $\Omega_1$, respectively. Therefore, we can compute the value of $p_\alpha(\omega_j)$ by solving the following equation: $$
\begin{array}{lll}
\beta & = & \sum_{i = 1}^{j-1} \mu(\omega_i)(\alpha v(0, \omega_i) \\
& + & (1-\alpha) v(1, \omega_i)) \\
& + & p_\alpha(\omega_j)\mu(\omega_j)v(0, \omega_j)\\ 
& + & \sum_{\omega \in \Omega_{j,0}} \mu(\omega)v(0, \omega) \\
& + &\sum_{\omega \in \Omega_{j,1}} \mu(\omega)v(1, \omega),
\end{array}
$$ where $\Omega_{j,0} := (\Omega \setminus \{\omega_1, \ldots, \omega_j\}) \cap \Omega_0$ and $\Omega_{j,1} := (\Omega \setminus \{\omega_1, \ldots, \omega_j\}) \cap \Omega_1$. Note that this means that $p_\alpha(\omega_j)$ is locally linear in $\alpha$, and thus that the senders' expected utility given by $p_\alpha$ is also locally linear in $\alpha$ whenever $p_\alpha(\omega_j)$ attains values outside of $\{0, \alpha, 1\}$. Therefore, the senders' expected utility given by $p_\alpha$ is piecewise linear as a function of $\alpha$, and hence its maximum lies in one of the segment's endpoints 
(i.e., in one of the values of $\alpha$ such that $p_\alpha(\omega_j) \in \{0,1,\alpha\}$ for all $\omega \in \Omega_D$). 
\end{proof}

The algorithm we provide to find the best equilibrium for the senders involves checking the senders' utilities at each of these endpoints. Note that a segment's endpoint is precisely a value of $\alpha$ such that the solution found by the above algorithm attains values only in $\{0, \alpha, 1\}$. By construction, the only possible preimages of $\alpha$ in $p_\alpha$ are $S_1 := \{\omega_1\}$, $S_2 := \{\omega_1, \omega_2\}$, $\ldots$, $S_k := \{\omega_1, \ldots, \omega_k\}$. Moreover, for each of these sets $S_j$, there is at most one value of $\alpha_j \in [0,1]$ such that $p_{\alpha_j}^{-1}(\alpha_j) = S_j$, and $\alpha_j$ is given by  $$\beta = \sum_{i = 1}^{j} \mu(\omega_i)(\alpha_j v(0, \omega_i) + (1-\alpha_j) v(1, \omega_i)) + $$ $$ + \sum_{\omega \in \Omega_{j,0}} \mu(\omega)v(0, \omega)  + \sum_{\omega \in \Omega_{j,1}} \mu(\omega)v(1, \omega).$$ Isolating $\alpha_j$ from the equation we get 

$$\alpha_j = \frac{Z_j}{\sum_{i = 1}^j\mu(\omega_i)(v(0, \omega_i) - v(1, \omega_i))},$$

where $$
\begin{array}{lll}
Z_j & = & \beta - \sum_{i = 1}^j\mu(\omega_i)v(1, \omega_i) \\
& - & \sum_{\omega \in \Omega_{j,0}} \mu(\omega)v(0, \omega) \\
& - & \sum_{\omega \in \Omega_{j,1}} \mu(\omega)v(1, \omega)
\end{array}$$


Putting everything together, our algorithm is as follows:

\begin{enumerate}
    \item \textbf{Step 1:} Compute the best possible mechanism for the senders (i.e., set $p(\omega) = 1$ for $\omega \in \Omega_0$ and $p(\omega) = 0$ for $\omega \in \Omega_1$). If this mechanism gives the receiver a utility greater than or equal to $\beta$, return this configuration and terminate.
    \item \textbf{Step 2:} Compute $p_0$ and $p_1$ and set the best configuration $p$ to be the one between $p_0$ and $p_1$ that reports the most utility to the senders.
    \item \textbf{Step 3:} For $j = 1,2,\ldots,k$, compute $\alpha_j$. If $\alpha_j \in [0,1]$ and $p_{\alpha_j}$ is better for the senders than $p$, set $p$ to $p_{\alpha_j}$.
    Return $p$.
\end{enumerate}

Note that if $\Omega_D$ is sorted by resistance beforehand, this algorithm takes $O(n)$ operations to compute $p$ since all the partial sums used to compute $\alpha_j$ and the senders' utilities can either be precalculated in $O(n)$ operations, or they can simply be updated by adding one term to each sum at each iteration $j$. If $\Omega_D$ is not sorted, the algorithm takes $O(n \log n)$ operations since it needs to sort the states by resistance first.

We can now apply the algorithm to Example \ref{example} above and get that the optimal policy for the sender is $p=(0,1,\frac{10}{19},0,0,0)$. Interestingly, this policy yields a utility of $\frac{43}{57}\approx 0.754$ for the senders, as opposed to  $\frac{91}{120}\approx 0.758$ in the single sender Bayesian persuasion (with  commitment power). In contrast, a single sender's cheap talk equilibrium (with no commitment power) yields a utility of $\frac{1}{2}$ for the sender. An interesting follow-up question is what is the maximal (normalized) loss of the best cheap talk equilibrium over the Bayesian persuasion optimal revelation policy.  

\subsubsection{General Case}\label{sec:general-case-sketch}

In this section we sketch the construction of an $O(n \log n)$ algorithm that outputs the optimal equilibrium for the first sender in the general case, in which senders may not have common interests. The full construction can be found at Appendix~\ref{sec:general-algo}.

Given $\alpha, \gamma \in [0,1]$, consider the problem of finding the optimal policy $p_{\alpha, \gamma}$ for the first sender such that $p_{\alpha, \gamma}(\omega) = \alpha$ for all $\omega \in \Omega_{1,0}$ and $p_{\alpha, \gamma}(\omega) = \gamma$ for all $\omega \in \Omega_{0,1}$. By Theorem~\ref{th:main}, the policy $p_{\alpha, \gamma}$ that gives the most utility to the first sender is also the actual implementable policy that is optimal for the first sender (with no constraints). It is straightforward to check that $p_{\alpha, \gamma}$ can be computed with a slight variation of the algorithm that outputs $p_\alpha$ in the common interest case. Thus, it remains to check which values of $\alpha$ and $\gamma$ maximize the utility of the first sender.

A similar argument to the one used in the previous section
shows that there is at least one optimal policy in which $p_{\alpha, \gamma}(\omega) \in \{0, \alpha, \gamma, 1\}$ for all $\omega \in \Omega$. Unfortunately, by contrast to the common interest case, this still leaves us with an infinite number of possibilities for $\alpha$ and $\gamma$. In fact, it can be shown that, if we assume that $\alpha \le \gamma$, the possible solutions can be partitioned into several cases (at most $n$) in which a linear equation on $\alpha$ and $\gamma$ must be satisfied. Since the expected utility of the first sender is also linear in $\alpha$ and $\gamma$, for each of these cases there always exists an optimal solution in the boundary, which is when $\alpha = \gamma$ or $(\alpha, \gamma) \in \{(0,0), (0,1), (1,1)\}$. Considering also the cases in which $\alpha \ge \gamma$ gives us the additional solution where $\alpha = 1$ and $\gamma = 0$. This additional constraint allows us to reduce the number of possible optimal solutions to a finite number which is linear in $n$. Each of these solutions can be computed in constant time in the same fashion as in the common interest algorithm. 

\subsection{Best Receiver Equilibrium}
Most literature on Bayesian persuasion focuses on the best sender equilibrium. This is because the informed sender can also decide on how information is revealed. In our case it make sense to assume that the information designer that determines the equilibrium selection has the same incentive as the receiver. As we shall now show, determining the best receiver equilibrium is easy. 

Here we take the general approach where the preferences of the two senders may not be aligned. For simplicity, we consider the case where no sender is indifferent between the two actions in any state. In this case, we can use Corollary \ref{corollary:no ties} to determine the optimal policy. 
We call a policy $p$ \emph{pure} if $p:\Omega\to\Delta(A)$ is a Dirac measure on either $0$ or $1$. Let $\Omega_F\subseteq\Omega$ be the set of states where all three decision makers have the same preference.
Let $\mathcal{P}$ be the set of all pure policies such that (i) $p(\omega)$ recommends the commonly preferred action for every $\omega\in\Omega_F$ and (ii) every $p\in\mathcal{P}$ is constant across types of states for all four different types of states in $\Omega\setminus{\Omega_F}$ (recall Corollary \ref{corollary:no ties}). That is, $p(\omega)=p(\omega')$ for every $\omega,\omega'\in\Omega_{0,0}\setminus\Omega_F$, $p(\omega)=p(\omega')$ for every $\omega,\omega'\in\Omega_{1,0}$, etc. 

We claim that by the feasibility constraint $\mathcal{P}$ contains $6$ policies. To see this, note that since any $p\in\mathcal{P}$ is pure and  is fixed over states in $\Omega_F$, it can be described by a vector in $\{0,1\}^4$ according to its values in the four types of states: the first value represents its value in $\Omega_{0,0}$, the second represents its value in $\Omega_{1,0}$, the third represents its value in $\Omega_{0,1}$, and the fourth represents its value in $\Omega_{1,1}$. The feasibility constraint asserts that the first value must be the global maximum across the four values and the last value must be the global minimum. We note that the policy $(1,1,1,1)$  dominates the policy $(1,1,1,0)$ for the receiver. This is true since, by definition, in $\Omega_{1,1}^0$ the receiver is in disagreement with both senders and prefers action $0$ in all these states. Similarly, $(0,0,0,0)$ dominates $(1,0,0,0)$. We denote the remaining four policies, after omitting the two dominated policies by $\mathcal{P}^*$   

We have the following simple corollary to determine the optimal policy for the receiver.
\begin{lemma}
There exists an optimal policy for the receiver that lies in $\mathcal{P}^*$.
\end{lemma}
\begin{proof}
Consider first the set of all implementable policies for the sender. This set is a convex polytope in $\mathbb{R}^{\Omega}$.
The vertices of this polytope are pure policies. In addition, the utility of the receiver is linear over the polytope and so the maximum implementable policy for the receiver is attained as a pure policy. Moreover, we can assume that in $\Omega_F$ the recommended action is the consensus action. This is true since taking any policy and altering it in $\Omega_F$ according to the consensual action retains feasibility and increases the utility of the receiver. 

To complete the proof we only need to show why a policy $p$ that has both values $0$ and $1$ over states in $\Omega_{0,0}\setminus\Omega_F$ or two values over states in $\Omega_{1,1}\setminus\Omega_F$ can be improved by a policy in $\mathcal{P}$. To see this, note that if policy $p$ has both values $0$ and $1$  over states in $\Omega_{0,0}\setminus\Omega_F$, then by the feasibility constraints we must have that $p$ is $0$ across all states in $\Omega\setminus\Omega_F$ that are  not in $\Omega_{0,0}$. Since, by definition, the receiver prefers action $1$ in all states in $\Omega_{0,0} \setminus \Omega_F$, the policy $(0,0,0,0)$ in $\mathcal{P}^*$ dominates $p$. A similar consideration shows that $(1,1,1,1)$ dominates all implementable policies that attain two values over states in $\Omega_{1,1}\setminus\Omega_F$.      
\end{proof}

\section{Conclusion and Open Problems}
In this work, we characterized all incentive-compatible policies in a setting with two senders and one receiver with no commitment power, in which all agents can communicate through a trusted mediator. This characterization is also valid in the cheap talk setting, where there is no mediator and all agents can communicate with each other through private authenticated channels. However, in the cheap talk setting, the implementable policies in general allow a small probability of error. We also provided an $O(n \log n)$ algorithm (where $n$ is the number of states) that finds the optimal policy for each of the senders, and a very simple mechanism that is optimal for the receiver. 

Our results show that when there are two senders the equilibrium outcomes are much richer and are closer to those of classical Bayesian persuasion but without the commitment power assumption. A natural question to ask is whether our results can be extended to a more general setting. In particular, it is still open whether one can find a similar characterization in the following settings:

\begin{itemize}
    \item [(a)] A setting in which the receiver can play more than two actions.
    \item [(b)] A setting in which there are more than two senders, but each of them possesses only partial information about the state. Note that, if there are more than two senders and all of them are fully informed, then any policy $p$ can be implemented by the mediator. This can be done by setting the set of messages equal to the set of states $\Omega$. In this way, the mediator can fully deduce the state $\omega$ by taking the majority between the messages sent by the senders and then sampling an action from $p(\omega)$.
    Thus, for more than two senders the setting is interesting only if they are not fully aware of the state. 
    \item [(c)] A setting in which there are more than two senders, but up to $k$ of them can collude and deviate from the proposed strategy in a coordinated way.
\end{itemize}

\bibliography{bibfile}

\begin{thebibliography}{20}
\providecommand{\natexlab}[1]{#1}

\bibitem[{Abraham et~al.(2006)Abraham, Dolev, Gonen, and Halpern}]{ADGH06}
Abraham, I.; Dolev, D.; Gonen, R.; and Halpern, J.~Y. 2006.
\newblock Distributed computing meets game theory: robust mechanisms for
  rational secret sharing and multiparty computation.
\newblock 53--62.

\bibitem[{Akerlof(1978)}]{akerlof1978market}
Akerlof, G.~A. 1978.
\newblock The market for “lemons”: Quality uncertainty and the market
  mechanism.
\newblock In \emph{Uncertainty in economics}, 235--251. Elsevier.

\bibitem[{Arieli, Babichenko, and Sandomirskiy(2022)}]{arieli2022bayesian}
Arieli, I.; Babichenko, Y.; and Sandomirskiy, F. 2022.
\newblock Bayesian Persuasion with Mediators.
\newblock \emph{arXiv preprint arXiv:2203.04285}.

\bibitem[{Arieli and Danino(2019)}]{arieli2019delegated}
Arieli, I.; and Danino, G. 2019.
\newblock Delegated persuasion.
\newblock \emph{Available at SSRN 3421954}.

\bibitem[{Aumann(1987)}]{aumann1987correlated}
Aumann, R.~J. 1987.
\newblock Correlated equilibrium as an expression of Bayesian rationality.
\newblock \emph{Econometrica: Journal of the Econometric Society}, 1--18.

\bibitem[{Bahar, Smorodinsky, and Tennenholtz(2019)}]{BaharST19}
Bahar, G.; Smorodinsky, R.; and Tennenholtz, M. 2019.
\newblock Social Learning and the Innkeeper's Challenge.
\newblock In Karlin, A.; Immorlica, N.; and Johari, R., eds., \emph{Proceedings
  of the 2019 {ACM} Conference on Economics and Computation, {EC} 2019,
  Phoenix, AZ, USA, June 24-28, 2019}, 153--170. {ACM}.

\bibitem[{Conitzer and Sandholm(2006)}]{Conitzer2006-ur}
Conitzer, V.; and Sandholm, T. 2006.
\newblock Computing the optimal strategy to commit to.
\newblock In \emph{Proceedings of the 7th {ACM} conference on Electronic
  commerce - {EC} '06}. New York, USA: ACM Press.

\bibitem[{Crawford and Sobel(1982)}]{crawford1982strategic}
Crawford, V.~P.; and Sobel, J. 1982.
\newblock Strategic information transmission.
\newblock \emph{Econometrica: Journal of the Econometric Society}, 1431--1451.

\bibitem[{Emek et~al.(2012)Emek, Feldman, Gamzu, Leme, and
  Tennenholtz}]{EmekFGLT12}
Emek, Y.; Feldman, M.; Gamzu, I.; Leme, R.~P.; and Tennenholtz, M. 2012.
\newblock Signaling schemes for revenue maximization.
\newblock In Faltings, B.; Leyton{-}Brown, K.; and Ipeirotis, P., eds.,
  \emph{Proceedings of the 13th {ACM} Conference on Electronic Commerce, {EC}
  2012, Valencia, Spain, June 4-8, 2012}, 514--531. {ACM}.

\bibitem[{Fujii and Sakaue(2022)}]{Fujii2022-aa}
Fujii, K.; and Sakaue, S. 2022.
\newblock Algorithmic Bayesian persuasion with combinatorial actions.
\newblock \emph{Proceedings of the AAAI Conference on Artificial Intelligence},
  36(5): 5016--5024.

\bibitem[{Gan et~al.(2022)Gan, Majumdar, Radanovic, and Singla}]{Gan2022-gm}
Gan, J.; Majumdar, R.; Radanovic, G.; and Singla, A. 2022.
\newblock Bayesian persuasion in sequential decision-making.
\newblock \emph{Proceedings of the AAAI Conference on Artificial Intelligence},
  36(5): 5025--5033.

\bibitem[{Kamenica and Gentzkow(2011)}]{kamenica2011bayesian}
Kamenica, E.; and Gentzkow, M. 2011.
\newblock Bayesian persuasion.
\newblock \emph{American Economic Review}, 101(6): 2590--2615.

\bibitem[{Kamenica and Gentzkow(2017)}]{kamenica2017competition}
Kamenica, E.; and Gentzkow, M. 2017.
\newblock Competition in persuasion.
\newblock \emph{Review of economic studies}, 84(1): 1.

\bibitem[{Kosenko(2018)}]{kosenko2018mediated}
Kosenko, A. 2018.
\newblock Mediated persuasion.
\newblock \emph{Available at SSRN 3276453}.

\bibitem[{Kremer, Mansour, and Perry(2013)}]{KremerMP13}
Kremer, I.; Mansour, Y.; and Perry, M. 2013.
\newblock Implementing the "Wisdom of the Crowd".
\newblock In Kearns, M.~J.; McAfee, R.~P.; and Tardos, {\'{E}}., eds.,
  \emph{Proceedings of the fourteenth {ACM} Conference on Electronic Commerce,
  {EC} 2013, Philadelphia, PA, USA, June 16-20, 2013}, 605--606. {ACM}.

\bibitem[{Krishna and Morgan(2001)}]{krishna2001model}
Krishna, V.; and Morgan, J. 2001.
\newblock A model of expertise.
\newblock \emph{The Quarterly Journal of Economics}, 116(2): 747--775.

\bibitem[{Lipnowski and Ravid(2020)}]{lipnowski2020cheap}
Lipnowski, E.; and Ravid, D. 2020.
\newblock Cheap talk with transparent motives.
\newblock \emph{Econometrica}, 88(4): 1631--1660.

\bibitem[{Morgan and Morrison(1999)}]{morgan1999models}
Morgan, M.~S.; and Morrison, M. 1999.
\newblock \emph{Models as mediators}.
\newblock Cambridge University Press Cambridge.

\bibitem[{Renault, Solan, and Vieille(2017)}]{renault2017optimal}
Renault, J.; Solan, E.; and Vieille, N. 2017.
\newblock Optimal dynamic information provision.
\newblock \emph{Games and Economic Behavior}, 104: 329--349.

\bibitem[{Salamanca(2021)}]{salamanca2021value}
Salamanca, A. 2021.
\newblock The value of mediated communication.
\newblock \emph{Journal of Economic Theory}, 192: 105191.

\end{thebibliography}

\appendix

\section{Best Sender Equilibrium - General Case}\label{sec:general-algo}

In this section we provide an $O(n \log n)$ algorithm that outputs the best equilibrium for the first sender in the general case. Let $\Omega_{a,b}$ be the set of states in which sender 1 prefers action $a$ and sender 2 prefers action $b$. Additionally, we define by $\Omega_{a,b}^c$ the subset of $\Omega_{a,b}$ in which the receiver prefers action $c$. By Theorem~\ref{th:main}, we know that all implementable solutions $p$ must satisfy the following constraints:
\begin{itemize}
    \item If $\omega \in \Omega_{0,0}$ and $\omega' \not \in \Omega_{0,0}$ then, $p(\omega) \ge p(\omega')$.
    \item If $\omega \in \Omega_{1,1}$ and $\omega' \not \in \Omega_{1,1}$ then, $p(\omega) \le p(\omega')$.
    \item If $\omega, \omega' \in \Omega_{1,0}$, then $p(\omega) = p(\omega')$.
    \item If $\omega, \omega' \in \Omega_{0,1}$, then $p(\omega) = p(\omega')$.
\end{itemize}

As in Section~\ref{sec:common-interest-algo}, consider the problem of computing the optimal implementable policy for the first sender $p_{\alpha, \gamma}$ in which $p_{\alpha, \gamma}(\omega) = \alpha$ for all $\omega \in \Omega_{1,0}$ and $p_{\alpha, \gamma}(\omega) = \gamma$ for all $\omega \in \Omega_{0,1}$. By Theorem~\ref{th:main}, the policy $p_{\alpha, \gamma}$ that maximizes the utility of the first sender is also the actual implementable policy that is optimal for the first sender (with no constraints). An argument analogous to the one applied for $p_{\alpha}$ in Section~\ref{sec:common-interest-algo} shows that, when it exists, $p_{\alpha, \gamma}$ can be computed using the following algorithm:

\begin{enumerate}
    \item \textbf{Step 1:} To each state $\omega \in \Omega$, assign $p_{\alpha, \gamma}(\omega) = 1$ if $\omega \in \Omega_{0,0}$; $p_{\alpha, \gamma}(\omega) = 0$ if $\omega \in \Omega_{1,1}$; $p_{\alpha, \gamma}(\omega) = \alpha$ if $\omega \in \Omega_{1,0}$; and $p_{\alpha, \gamma}(\omega) = \gamma$ if $\omega \in \Omega_{0,1}$. If the receiver's utility this way is greater than $\beta$, return this configuration and terminate.
    \item \textbf{Step 2:} Let $\Omega_D = \Omega_{0,0}^1 \cup \Omega_{1,1}^0$ be the set of states in which both senders have the same preferences but in which they disagree with the receiver. Sort the elements $\omega$ of $\Omega_D$ according to their resistance $r(\omega) := \frac{v(1,\omega) - v(0, \omega)}{u_1(0,\omega) - u_1(1, \omega)}$, and let $\{\omega_1, \ldots, \omega_k\}$ be the sorted elements of $\Omega_D$.
    \item \textbf{Step 3:} Iterate through $\omega_1, \ldots, \omega_k$. If $p_{\alpha, \gamma}(\omega_i) = 1$, decrease this value until either $p_{\alpha, \gamma} = \max(\alpha, \gamma)$ or the receiver's utility equals $\beta$. If $p_{\alpha, \gamma}(\omega_i) = 0$, increase this value until either $p_{\alpha, \gamma}(\omega_i) = \min(\alpha, \gamma)$ or the receiver's utility equals $\beta$. After each step, if the utility of the receiver is equal to $\beta$, return this configuration and terminate.
    \item \textbf{Step 4:} If no solution was found in Step 3 (i.e., if after iterating through all $k$ elements the utility of the sender never reached $\beta$), there is no incentive-compatible configuration for $\alpha$ and $\gamma$.
\end{enumerate}

Again, a similar argument to the one provided in Section~\ref{sec:common-interest-algo} shows that the first sender's utility with $p_{\alpha, \gamma}$ is piecewise linear in $\alpha$ and $\gamma$, and thus there exists an optimal implementable policy for the first sender in which all $p_{\alpha, \gamma}(\omega) \in \{0,\alpha, \gamma, 1\}$ for all $\omega \in \Omega$. We can bound the domain of $\alpha$ and $\gamma$ using the following lemma:

\begin{lemma}\label{lemma:char-general}
There exists an optimal implementable policy $p_{\alpha, \gamma}$ for the first sender in which $p_{\alpha, \gamma}(\omega) \in \{0,\alpha, \gamma, 1\}$ for all $\omega \in \Omega$ and one of the following conditions hold:
\begin{itemize}
    \item (a) $\min(\alpha, \gamma) = 0$.
    \item (b) $\max(\alpha, \gamma) = 1$.
    \item (c) $\alpha = \gamma$.
\end{itemize}
\end{lemma}

\begin{proof}
By the above argument, we can restrict the search for optimal configurations to those policies $p_{\alpha, \gamma}$ such that  $p_{\alpha, \gamma}(\omega) \in \{0,\alpha, \gamma, 1\}$ for all $\omega \in \Omega$. Next we show that there exists one such policy in which (a), (b), or (c) holds.

Suppose that $\alpha \ge \gamma$. By construction of $p_{\alpha, \gamma}$, the only subset of elements that might take values different from $0$ and $1$ are precisely $\Omega_D^0 := \emptyset$, $\Omega_D^1 = \{\omega_1\}$, $\Omega_D^2 = \{\omega_1, \omega_2\}$, $\ldots$, $\Omega_D^k = \Omega_D$. Moreover, for each such subset $\Omega_D^j$, if $p_{\alpha, \gamma}(\omega) \in \{\alpha, \gamma\}$ for all $\omega \in \Omega_D^j$ and $p_{\alpha, \gamma}(\omega) \in \{\alpha, \gamma\} \in \{0,1\}$ for all $\omega \not \in \Omega_D^j$, then the following equation must hold:

\begin{equation}\label{eq:condition-receiver}
\begin{array}{lll}
\beta  &  =  & \sum_{\omega \in \Omega_{0,0} \setminus \Omega_D^j} \mu(\omega) v(0,\omega) \\
& + & \sum_{\omega \in \Omega_{1,1} \setminus \Omega_D^j}  \mu(\omega) v(1,\omega) \\
& + & \sum_{\omega \in \Omega_{1,0} \cup (\Omega_D^j \cap \Omega_{0,0}) }  
\\
& & \mu(\omega) (\alpha \cdot v(0,\omega) + (1 - \alpha)v(1,\omega)) \\
& + &  \sum_{\omega \in \Omega_{0, 1} \cup (\Omega_D^j \cap \Omega_{1,1})} \\
& & \mu(\omega) (\gamma \cdot v(0,\omega) + (1 - \gamma)v(1,\omega)).\\
\end{array}
\end{equation}

This equation states that the expected utility that the receiver gets by setting $p_{\alpha, \gamma}(\omega)$ to $0,1, \alpha$, or $\gamma$ according to the above algorithm should be exactly $\beta$, which by Lemmas~\ref{lem:char} and \ref{lemma:beta-utility} is a necessary condition for $p_{\alpha, \gamma}$ to be optimal. Consider the following notation:

$$
\begin{array}{lll}
B_j  &  :=  & \beta - \sum_{\omega \in \Omega_{0,0} \setminus \Omega_D^j} \mu(\omega) v(0,\omega) \\
& - & \sum_{\omega \in \Omega_{1,1} \setminus \Omega_D^j}  \mu(\omega) v(1,\omega) \\
& - & \sum_{\omega \in \Omega_{1,0} \cup (\Omega_D^j \cap \Omega_{0,0}) }  \mu(\omega)v(1,\omega) \\
& - &  \sum_{\omega \in \Omega_{0, 1} \cup (\Omega_D^j \cap \Omega_{1,1})}  \mu(\omega) v(1,\omega),\\
\end{array}
$$

$$
\begin{array}{lll}
B'_j  &  :=  &  \sum_{\omega \in \Omega_{0,0} \setminus \Omega_D^j} \mu(\omega) u_1(0,\omega) \\
& + & \sum_{\omega \in \Omega_{1,1} \setminus \Omega_D^j}  \mu(\omega) u_1(1,\omega) \\
& + & \sum_{\omega \in \Omega_{1,0} \cup (\Omega_D^j \cap \Omega_{0,0}) }  \mu(\omega)u_1(1,\omega) \\
& + &  \sum_{\omega \in \Omega_{0, 1} \cup (\Omega_D^j \cap \Omega_{1,1})}  \mu(\omega) u_1(1,\omega),\\
\end{array}
$$

$$
A_j = \sum_{\omega \in \Omega_{1,0} \cup (\Omega_D^j \cap \Omega_{0,0}) }  \mu(\omega)(v(0,\omega) - v(1, \omega))
$$

$$
A'_j = \sum_{\omega \in \Omega_{1,0} \cup (\Omega_D^j \cap \Omega_{0,0}) }  \mu(\omega)(u(0,\omega) - u_1(1, \omega))
$$

$$
C_j = \sum_{\omega \in \Omega_{0,1} \cup (\Omega_D^j \cap \Omega_{1,1}) }  \mu(\omega)(v(0,\omega) - v(1, \omega))
$$

$$
C'_j = \sum_{\omega \in \Omega_{0,1} \cup (\Omega_D^j \cap \Omega_{1,1}) }  \mu(\omega)(u(0,\omega) - u_1(1, \omega)).
$$

The problem of maximizing the first sender's utility while satisfying Equation~\ref{eq:condition-receiver} can be formulated as one of maximizing $A'_j \alpha + C'_j \gamma + B'_j$ constrained to $A_j \alpha + C_j \gamma = B_j$ and $0 \le \gamma \le \alpha \le 1$. Since both equations are linear in $\alpha$ and $\gamma$, there exists a solution in the boundary of the domain, i.e., when either $\gamma = 0$, $\alpha = 1$, or $\alpha = \gamma$. The case in which $\gamma \ge \alpha$ also gives possible solutions when $\alpha = 0$ or $\gamma = 1$. This proves Lemma~\ref{lemma:char-general}.
\end{proof}

Lemma~\ref{lemma:char-general} shows that we can limit our search only to the values of $\alpha$ and $\gamma$ that satisfy either $\alpha = \gamma$ or $\alpha,\gamma \in \{0,1\}$. Thus, the following algorithm can find the optimal implementable policy $p$ in $O(n \log n)$ operations:

\begin{enumerate}
    \item \textbf{Step 1:} Compute the best possible mechanism for the first sender (i.e., set $p(\omega) = 1$ for $\omega \in \Omega_{0,1} \cup \Omega_{0,0}$ and $p(\omega) = 0$ for $\omega \in \Omega_{1,0} \cup \Omega_{1,1}$). If this mechanism gives the receiver a utility greater than or equal to $\beta$, return this configuration and terminate.
    \item \textbf{Step 2:} For each $j = 0,1,\ldots, k$, solve $A_j\alpha_j + C_j \gamma_j = B_j$ for (a) $\alpha_j = 0$, (b) $\alpha_j = 1$, (c) $\gamma_j = 0$, (d) $\gamma_j = 1$, and (e) $\alpha_j = \gamma_j$.
    \item \textbf{Step 3:} For each pair of $(\alpha_j, \gamma_j)$ found in Step 2, return the policy $p_{\alpha_j, \gamma_j}$ that maximizes $A'_j \alpha_j + C'_j \gamma_j + B'_j$.
\end{enumerate}

As in Section~\ref{sec:common-interest-algo}, all the partial sums used in the algorithm can be precomputed in linear time, and thus the algorithm takes at most $O(n \log n)$ operations since there is an overhead of $O(n \log n)$ operations required to sort the elements of $\Omega_D$.

\end{document}